\documentclass[11pt,a4paper]{article}
\usepackage[utf8]{inputenc}
\usepackage{amssymb}
\usepackage{amsthm}
\usepackage{mathtools}
\usepackage{color}
\usepackage{algorithm,algpseudocode}
\usepackage{booktabs}
\usepackage{geometry}
\usepackage{hyperref}
\PassOptionsToPackage{hyphens}{url}\usepackage{hyperref}
\usepackage[hyphenbreaks]{breakurl}

\title{Linearly Self-Equivalent APN Permutations in Small Dimension\thanks{This  work  was  funded  by Deutsche  Forschungsgemeinschaft  (DFG);  project  number 411879806 and by DFG under Germany's Excellence Strategy - EXC 2092 CASA - 390781972.}\thanks{This is the version accepted to IEEE Transactions on Information Theory. DOI of the final published version: \href{https://dx.doi.org/10.1109/TIT.2021.3071533}{10.1109/TIT.2021.3071533}.}\thanks{
$\copyright$ 2021 IEEE.  Personal use of this material is permitted.  Permission from IEEE must be obtained for all other uses, in any current or future media, including reprinting/republishing this material for advertising or promotional purposes, creating new collective works, for resale or redistribution to servers or lists, or reuse of any copyrighted component of this work in other works.}}
\renewcommand\footnotemark{}
\author{Christof Beierle, Marcus Brinkmann, Gregor Leander \\ \\
Ruhr University Bochum, Bochum, Germany}
\date{}

\newcommand{\F}{\mathbb{F}}
\newcommand{\id}{\mathrm{id}}
\newcommand{\Aut}{\mathsf{Aut}}
\newcommand{\AutAE}{\mathsf{Aut_{AE}}}
\newcommand{\AutLE}{\mathsf{Aut_{LE}}}
\newcommand{\GL}{\mathrm{GL}}
\newcommand{\AGL}{\mathrm{AGL}}
\newcommand{\Ord}{\mathrm{Ord}}
\newcommand{\ord}{\mathrm{ord}}
\newcommand{\comm}{\mathrm{Comm}}
\newcommand{\rcf}{\mathrm{RCF}}
\newcommand{\lcm}{\mathrm{lcm}}
\newcommand{\tr}{\mathrm{Tr}}
\newcommand{\wt}{\mathrm{wt}}
\newcommand{\companion}{\mathrm{Comp}}

\newtheorem{definition}{Definition}
\newtheorem{lemma}{Lemma}
\newtheorem{proposition}{Proposition}
\newtheorem{theorem}{Theorem}
\newtheorem{corollary}{Corollary}
\newtheorem{conjecture}{Conjecture}
\theoremstyle{remark}
\newtheorem{example}{Example}
\newtheorem{remark}{Remark}

\begin{document}

\maketitle

\begin{abstract}
    All almost perfect nonlinear (APN) permutations that we know to date admit a special kind of linear self-equivalence, i.e., there exists a permutation $G$ in their CCZ-equivalence class and two linear permutations $A$ and $B$, such that $G \circ A = B \circ G$. After providing a survey on the known APN functions with a focus on the existence of self-equivalences, we search for APN permutations in dimension 6, 7, and 8 that admit such a linear self-equivalence. In dimension six, we were able to conduct an \emph{exhaustive search} and obtain that there is only one such APN permutation up to CCZ-equivalence. In dimensions 7 and 8, we performed an exhaustive search for all but a few classes of linear self-equivalences and we did not find any new APN permutation. As one interesting result in dimension 7, we obtain that all APN permutation polynomials with coefficients in $\F_2$ must be (up to CCZ-equivalence) monomial functions.
    
\end{abstract}

\section{Introduction}
Differential cryptanalysis~\cite{DBLP:journals/joc/BihamS91} certainly belongs to the most important attack vectors to consider when designing a new symmetric cryptographic primitive. The basic idea of this attack is that the adversary chooses an input difference $a$ in the plaintext space and evaluates the encryption of pairs of values $(x,x +a)$ for a plaintext $x$. The goal is to predict the output difference of the two ciphertexts with a high probability. Vectorial Boolean functions (also known as \emph{S-boxes}) that offer the best resistance against differential attacks are called \emph{almost perfect nonlinear (APN)}. More precisely, a function $F \colon \F_2^n \rightarrow \F_2^m$ is called APN if, for every $b \in \F_2^m$ and non-zero $a \in \F_2^n$, the equation $F(x) + F(x + a) = b$ has at most two solutions. We know some instances and constructions of APN functions. However, much less is known if we require $F$ to be a permutation. 

For odd values of $n$, we know infinite families of APN permutations and in particular, APN permutations exist for every odd value of $n$. For even values of $n$ we only know one sporadic example up to CCZ-equivalence, which is defined for $n=6$ (see~\cite{browning2010apn}). We refer to this permutation as ``Dillon's permutation'' henceforth. Exhaustive search for APN permutations is only possible as long as $n$ is small because the size of the search space increases rapidly. So far, an exhaustive search for APN permutations has only been conducted up to $n=5$ (see~\cite{DBLP:journals/dcc/BrinkmannL08}). Since already for $n=6$, the number of permutations in $\F_2^n$ is orders of magnitude higher than for $n=5$ (i.e., $64!\approx 2^{296}$ for $n=6$ compared to $32!\approx 2^{117.7}$ for $n=5$), the search space has to be restricted. Although restricting the search space is a natural idea, the question remains how to do so. Our idea is to only consider the class of permutations that we conjecture to contain all possible cases. Namely, we restrict to the class of permutations that admit a non-trivial \emph{linear self-equivalence}, i.e., those permutations $F$ for which there exist non-trivial linear permutations $A$ and $B$ such that $F \circ A = B \circ F$. Indeed, we observe that all known APN permutations admit, up to CCZ-equivalence, such a non-trivial linear self-equivalence.

\subsection{Our Contribution}
In the first part of this work, we provide a survey on all APN functions known from the literature and observe that they all admit a non-trivial automorphism. An \emph{automorphism} of a vectorial Boolean function $F \colon \F_2^n \rightarrow \F_2^m$ is an affine permutation in $\F_2^{n} \times \F_2^m$ that leaves the set $\{(x,F(x)) \mid x \in \F_2^n\}$ invariant. For all the known APN permutations $F \colon \F_2^n \rightarrow \F_2^n$, we show that there exists an automorphism of a special kind, i.e., there exists a permutation $G$ which is CCZ-equivalent to $F$ that admits a non-trivial linear self-equivalence. Since a linear self-equivalence is a special kind of automorphism, we also call it an \emph{LE-automorphism}. We conjecture that the CCZ-equivalence class of any APN permutation contains a permutation with a non-trivial LE-automorphism (Conjecture~\ref{conj}).

Based on this conjecture, our goal is to conduct an exhaustive search for all such APN permutations in small dimension. As a first step, we classify all possible LE-automorphisms that need to be considered in such a search. Since an LE-automorphism can be represented by a tuple $(A,B)$, where $A$ and $B$ are invertible matrices with coefficients in $\F_2$, we need to reduce the number of such tuples in order to conduct an efficient search. The most important observation here is that we only need to consider matrices up to similarity and identical cycle type. Surprisingly, our reductions lead to a very small number of tuples, i.e., $17$ for $n=6$, $27$ for $n=7$, and $32$ for $n=8$. We stress that this reduction is valid for any kind of search among the permutations with non-trivial LE-automorphisms; it is not restricted to APN permutations.

We then use this classification of LE-automorphisms to search for APN permutations in dimension $n \in \{6,7,8\}$. By using the APN property, we can exclude some of the LE-automorphisms immediately (Propositions~\ref{prop:inv} and~\ref{prop:poly}). We handled the rest using a recursive tree search (Algorithm~\ref{alg:search}). For $n=6$, we were able to conduct an \emph{exhaustive} search for the APN permutations with non-trivial LE-automorphisms and conclude that only Dillon's permutation remains. In other words, if Conjecture~\ref{conj} is true, this would be the only APN permutation in dimension $6$ up to CCZ-equivalence. For $n=7$, we found all the APN monomial permutations, but no more CCZ-equivalence classes.  For $n=8$, we did not find APN permutations. Since we computationally handled all but a few missing cases of possible LE-automorphisms for $n \in \{7,8\}$, we conclude that if new CCZ-equivalence classes of APN permutations with non-trivial LE-automorphisms exist, those automorphisms have to be of a special form (Theorem~\ref{thm:res7} and~\ref{thm:res8}).   Our search for $n=7$ completely covers the special case of \emph{shift-invariant permutations}, which correspond to all permutation polynomials in $\F_{2^7}$ with coefficients in $\F_2$ (there are 20,851,424,802,623,573,443,244,703,744,000 of those, see ~\cite{carlitz1972permutations} and OEIS sequence A326932~\cite{oeis}). We obtain that the only shift-invariant APN permutations in  dimension $7$ are monomial functions.

Our results can be useful for future searches for APN permutations. In particular, Theorem~\ref{thm:res7} and Theorem~\ref{thm:res8} provide significantly reduced search spaces for new APN permutations in dimension $n=7$ and $n=8$. If Conjecture~\ref{conj} holds, taking the known APN monomial permutations and conducting an exhaustive search over those remaining search spaces would be sufficient in order to classify all APN permutations in dimension $n \in \{7,8\}$ up to CCZ-equivalence.

\section{Preliminaries}
We denote by $\mathbb{N} = \{1,2,3,\dots\}$ the set of natural numbers and by $\mathbb{Z}$ the set of integers. Let $\F_2 = \{0,1\}$ be the field with two elements and, for $n \in \mathbb{N}$, let $\F_2^n$ be the $n$-dimensional vector space over $\F_2$. Let $\GL(n,\F_2)$ denote the group of invertible $n \times n$ matrices over $\F_2$ and let $\AGL(n,\F_2)$ denote the group of affine permutations on $\F_2^n$, i.e., the set of functions of the form $x \mapsto Lx + b$ for $L \in \GL(n,\F_2)$ and $b \in \F_2^n$. Throughout this paper, we will use matrices and the linear functions that they represent interchangeably. In other words, for an $n \times n$ matrix $L$ over $\F_2$, we simply use the symbol $L$ for the function $x \mapsto Lx$. We denote by $I_n$ the identity matrix in $\GL(n,\F_2)$. We denote a block-diagonal matrix consisting of blocks $M_1,M_2,\dots,M_k$ by $M_1 \oplus M_2 \oplus \dots \oplus M_k$, where $M_1$ corresponds to the block in the upper left corner of the block-diagonal matrix. For a matrix $M \in \GL(n,\F_2)$, we denote by $\ord(M)$ the \emph{multiplicative order} of $M$, i.e., the smallest positive integer $i$ such that $M^{i} = I_n$. Similarly, for a vector $x \in \F_2^n$, we denote by $\ord_M(x)$ the smallest positive integer $i$ for which $M^i(x) = x$. It is well known that $\ord_M(x) \mid \ord(M)$. Indeed, suppose that $\ord(M) = r + k \cdot \ord_M(x)$ for a non-negative integer $r < \ord_M(x)$. We then have $x = M^{\ord(M)}(x) = M^{r+k \cdot \ord_M(x)}(x) = M^r(x)$, implying that $r=0$. Recall also that the \emph{minimal polynomial} of a matrix $M$ over $\F_2$ is the polynomial $p \in \F_2[X]$ of least degree such that $p(M)=0$.

For a polynomial 
$q = X^n + q_{n-1}X^{n-1} + \dots + q_1X + q_0 \in \mathbb{F}_2[X]$, the \emph{companion matrix} of $q$ is defined as the $n \times n$ matrix
\[
\companion(q) \coloneqq \left[\begin{array}{ccccc}
0 &  & &  & q_0 \\
1 & 0 &  & & q_1 \\
 & \ddots & \ddots &  & \vdots \\
 &  & 1 & 0 & q_{n-2} \\
 &  &  & 1 & q_{n-1}
\end{array}\right].
\]
The matrix $\companion(q)$ is invertible if and only if $q_0=1$.

A vectorial Boolean function $F \colon \F_2^n \rightarrow \F_2^n$ can be uniquely expressed as a multivariate polynomial in $\F_2^n[X_1,\dots,X_n]/(X_1^2+X_1,\dots,X_n^2+X_n)$ via the \emph{algebraic normal form}:
\[ F(x_1,\dots,x_n) = \sum_{u \in \F_2^n} \left( a_u \prod_{i \in \{1,\dots,n\}}x_i^{u_i}\right), \quad a_u \in \F_2^n.\]
 The \emph{algebraic degree} of $F$ is defined as $\max\{\wt(u) \mid a_u \neq 0, u \in \F_2^n\}$ where $\wt(u)$ denotes the Hamming weight of $u$. We recall that the \emph{Hamming weight} of a vector $u \in \F_2^n$ is defined as the number of its non-zero coordinates. The function $F$ is called \emph{affine} if it is of algebraic degree at most 1 and it is called \emph{quadratic} if it is of algebraic degree 2. The finite field $\F_{2^n}$ with $2^n$ elements is isomorphic to $\F_2^n$ as a vector space and any function $F \colon \F_2^n \rightarrow \F_2^n$ can also be uniquely represented as a univariate polynomial $f \coloneqq \sum_{i=0}^{2^n-1}\omega_iX^i$ in $\F_{2^n}[X]/(X
^{2^n}+X)$, which is called the \emph{univariate representation} of $F$. We then have
\[F\colon \F_{2^n} \rightarrow \F_{2^n}, x \mapsto f(x).\]

Since this paper focuses on APN functions, we first recall the definition.

\begin{definition}{\cite{DBLP:conf/crypto/NybergK92}}
A vectorial Boolean function $F\colon \F_2^n \rightarrow \F_2^n$ is called \emph{almost perfect nonlinear} (APN) if, for every $a \in \F_2^n\setminus \{0\} ,b \in \F_2^n$, the equation $F(x) + F(x+a) = b$ has at most 2 solutions for $x \in \F_2^n$.
\end{definition}
Let $F, G \colon \F_2^n \rightarrow \F_2^n$ be vectorial Boolean functions in dimension $n$. There are several well-known equivalence relations on vectorial Boolean functions that preserve the APN property and are therefore quite useful for classifying vectorial Boolean functions. We say that $G$ is \emph{linear equivalent} to $F$ if there exist $A,B \in \GL(n,\F_2)$ such that $F \circ A = B \circ G$. In the more general case when $A$ and $B$ are in $\AGL(n,\F_2)$, $G$ and $F$ are called \emph{affine equivalent}. We say that $G$ is \emph{extended-affine equivalent} (or \emph{EA-equivalent} for short) to $F$ if there exist $A,B \in \AGL(n,\F_2)$ and an affine, not necessarily invertible, function $C\colon \F_2
^{n} \mapsto \F_2^n$ such that $F \circ A = B \circ G + C$. Finally, we consider the notion of CCZ-equivalence~\cite{DBLP:journals/tit/BudaghyanCP06,ccz}. Let $\Gamma_F \coloneqq \{ (x,F(x))^{\top} \mid x \in \F_2^n\}$ be the \emph{graph} of $F$, where $(x,F(x))
^{\top}$ denotes the transposition of $(x,F(x))$ to a column vector of length $2n$. The functions $F$ and $G$ are called \emph{CCZ-equivalent} if there exists $\sigma \in \AGL(2n,\F_2)$ such that $\Gamma_G = \sigma(\Gamma_F)$. The above notions of equivalence are ordered by generality with CCZ-equivalence being the most general. In other words, two linear-equivalent functions are also (extended-) affine equivalent and two (extended-) affine-equivalent functions are also CCZ-equivalent. The latter can be seen by taking
\[ \sigma = \left[\begin{array}{cc}A^{-1} & 0\\ B^{-1}CA^{-1} & B^{-1}\end{array}\right]\in \AGL(2n,\F_2).\]
Note that CCZ-equivalence is strictly more general than EA-equivalence (combined with taking inverses in the case of permutations)~\cite{DBLP:journals/tit/BudaghyanCP06}.

The automorphism group (see~\cite[Definition 13]{DBLP:journals/ffa/CanteautP19}) of a function $F \colon \F_2^n \rightarrow \F_2^n$ is defined as
\[ \Aut(F) \coloneqq \{ \sigma \in \AGL(2n,\F_2) \mid \Gamma_F = \sigma(\Gamma_F)\}.\]
Analogously, we define the subgroups $\AutAE$ and $\AutLE$ as follows:
\begin{align*}
\AutAE(F) &\coloneqq \left\{ \sigma \in \Aut(F) \mid \sigma = \left[\begin{array}{cc}A & 0\\ 0 & B\end{array}\right] \text{ for } A,B \in \AGL(n,\F_2)\right\}, \\
\AutLE(F) &\coloneqq \left\{ \sigma \in \Aut(F) \mid \sigma = \left[\begin{array}{cc}A & 0\\ 0 & B\end{array}\right] \text{ for } A,B \in \GL(n,\F_2)\right\}.
\end{align*}

We have $\{\id\} \subseteq \AutLE(F) \subseteq \AutAE(F) \subseteq \Aut(F) \subseteq \AGL(2n,\F_2)$, where $\id$ denotes the identity mapping in $\AGL(2n,\F_2)$. The automorphism group of $F$, resp., the subgroups $\AutAE(F)$ and $\AutLE(F)$ contain non-trivial elements if and only if $F$ is self-equivalent with respect to the corresponding equivalence relation. For instance, \[\left[\begin{array}{cc}A & 0\\ 0 & B\end{array}\right] \in \Aut(F) \quad \Leftrightarrow \quad F \circ A = B \circ F .\] 
Self-equivalences of vectorial Boolean functions in small dimension have already been considered in the PhD thesis~\cite{de2007analysis}.
Note that if $F\colon \F_2^n \rightarrow \F_2^n$ and $G\colon \F_2^n \rightarrow \F_2^n$ are CCZ-equivalent, resp., affine equivalent, resp., linear equivalent, we have $\Aut(F) \cong \Aut(G)$, resp., $\AutAE(F) \cong \AutAE(G)$, resp., $\AutLE(F) \cong \AutLE(G)$, where $\cong$ denotes that the corresponding groups are isomorphic. Therefore, it is enough to consider only one single representative in each equivalence class when determining the automorphism groups. Throughout this paper, we are especially interested in APN permutations $F$ with a non-trivial subgroup $\AutLE(F)$. If $|\AutLE(F)| > 1$, we say that $F$ has a non-trivial \emph{LE-automorphism} (or a non-trivial \emph{linear self-equivalence}).

\section{Automorphisms of Some (APN) Functions}
It is well known (see~\cite{browning2009apn}) that the automorphism group of a function $F\colon \F_{2^n} \rightarrow \F_{2^n}$ can be computed by considering the associated linear code $C_F$ for $F$ with parity-check matrix
\[\left[\begin{array}{ccccc}1 & 1 & 1 &\dots & 1\\0 & a^0 & a^1 & \dots & a^{2^n-2}\\ F(0) & F(a^0) & F(a^1) & \dots & F(a^{2^n-2}) \end{array}\right],\]
where $a$ is a primitive element of $\F_{2^n}^*$,
and computing its automorphism group, e.g., with the algorithm presented in~\cite{feulner2009automorphism}. However, in this part of the paper we are only interested in determining whether $\Aut(F)$, resp., $\AutLE(F)$ contains a non-trivial element, especially for the case of an APN function $F$. In the following, we study some interesting classes of functions and sporadic APN instances.
\subsection{Quadratic Functions}
A lot of families of functions studied in the literature are CCZ-quadratic, i.e., they are CCZ-equivalent to a quadratic function. If quadratic functions are represented as mappings from the finite field $\F_{2^n}$ to itself, their polynomial representation only contains (besides a constant monomial in case that 0 is not mapped to 0) monomials of the form $x^{2^i}$ or $x^{2^i+2^j}$ with non-zero coefficients. One reason that they are of significant interest is because studying them is often easier than studying general functions. For instance, if $F\colon \F_2
^n \rightarrow \F_2^n$ is a quadratic function and $\alpha \in \F_2
^n$, the first-order derivative $x \mapsto F(x) + F(x+\alpha)$ is affine. The following is a well-known result (see e.g.~\cite[Proposition 1]{DBLP:journals/dcc/BrackenBMN11} or~\cite[Theorem 4]{DBLP:journals/amco/EdelP09}).
\begin{proposition}\label{prop:quadratic}
Let $F\colon \F_2^n \rightarrow \F_2^n$ be a quadratic function. Then, $|\Aut(F)| \geq 2^n$.
\end{proposition}
\begin{proof}
Let $\alpha \in \F_{2^n} \setminus \{0\}$. Since $F$ is quadratic, $F(x)+F(x+\alpha)+F(\alpha)+F(0)$ is linear. Let us denote this function by $L_{\alpha}$. The function defined as
\[\sigma_{\alpha}\colon \F_2^{2n} \rightarrow \F_2^{2n}, \quad (x,y) \mapsto (x+\alpha,y+L_{\alpha}(x)+F(\alpha)+L_{\alpha}(\alpha) + F(0))\]
belongs to $\Aut(F)$. Indeed, it is easy to see that $\sigma_{\alpha}$ is an affine permutation. We further have
\begin{align*}
    \sigma_{\alpha} (\Gamma_F) &= \sigma_{\alpha}\{(x,F(x))\} = \{(x+\alpha,F(x)+L_{\alpha}(x)+F(\alpha)+L_{\alpha}(\alpha)+F(0))\} \\
    &= \{(x,F(x+\alpha)+L_{\alpha}(x)+F(\alpha)+F(0))\} = \{(x, F(x))\} = \Gamma_F.
\end{align*}
This implies that $\Aut(F)$ contains $2^n-1$ non-trivial elements. Since $\id$ is trivially contained in $\Aut(F)$, this concludes  the proof.
\end{proof}

Many of the known APN functions are CCZ-quadratic, for instance the Gold functions $x \mapsto x^{2^i+1}$ for $\gcd(i,n)=1$ (see~\cite{DBLP:journals/tit/Gold68,DBLP:conf/eurocrypt/Nyberg93}), the two functions from $\F_{2^{10}}$ and $\F_{2^{12}}$ to itself defined in~\cite{DBLP:journals/tit/EdelKP06}, and the  classes defined in~\cite{DBLP:journals/ffa/BrackenBMM08,DBLP:journals/ccds/BrackenBMM11,cryptoeprint:2020:295,9000901,DBLP:journals/tit/BudaghyanC08,DBLP:journals/tit/BudaghyanCL08,DBLP:journals/ffa/BudaghyanCL09,5351383,DBLP:journals/iacr/BudaghyanHK19,DBLP:journals/dcc/Carlet11,taniguchi,ZHOU201343}. Indeed, to the best of our knowledge, at the time of writing only \emph{a single} APN function is known which is not CCZ-equivalent to either a quadratic or a monomial function (see~\cite{DBLP:journals/dcc/BrinkmannL08,DBLP:journals/amco/EdelP09}). We refer to~\cite[Sections 11.4--11.5]{carlet_book}, Table 1.6 in the PhD thesis~\cite{thesis_arshad}, and to~\cite{recent_list} for a recent summary of the known infinite classes of APN functions. In~\cite{recent_list}, the authors proved the CCZ-equivalence between some of the known families and thus provided a more reduced list.

We leave it as an open question whether for a (quadratic) function $F$ with a non-trivial automorphism in $\Aut(F)$ there always exists a representative $G$ in the CCZ-equivalence class of $F$ with a non-trivial automorphism in $\AutLE(G)$. However, we will see in the following that for many APN functions, such a representative does exist in the CCZ-equivalence class.

\subsection{Shift-Invariant Functions}
\label{sec:shift-inv}
The \emph{shift-invariant} functions $F \colon \F_2^n \rightarrow \F_2^n$ are exactly those for which 
\[\left[\begin{array}{cc}\companion(X^n+1) & 0\\ 0 & \companion(X^n+1)\end{array}\right] \in \AutLE(F).\]
 If $F$ is represented as a function from $\F_{2^n}$ to $\F_{2^n}$ in the normal basis representation, this automorphism corresponds to squaring, i.e., for all $x \in \F_{2^n}$, $F(x^2) = F(x)^2$ (see~\cite{normal_basis}). Therefore, the shift-invariant functions correspond to the polynomials with coefficients in $\F_2$.

Many of the known families of APN functions belong to the class of shift-invariant functions. Those include all the APN monomial functions $x \mapsto x^d$ and also the functions defined in ~\cite{DBLP:conf/waifi/Budaghyan07,DBLP:journals/ffa/BudaghyanCL09,DBLP:journals/tit/BudaghyanCP06}. For the latter, this is because $\tr_m(x^2)=\tr_m(x)^2$,
where $\tr_m(x)=\sum_{i=0}^{\frac{n}{m}-1}x^{2^{mi}}$ denotes the trace function from $\F_{2^n}$ to a subfield $\F_{2^m}$.

In~\cite{yu:eprint}, a classification of quadratic shift-invariant APN functions over $\F_2^n$ for $n \leq 9$ is provided.

\subsection{APN Binomial (and some Multinomial) Functions}
For monomial functions, a non-trivial LE-automorphism can easily be given in terms of multiplication with finite field elements. In particular, if $x \in \F_{2^n}$, then, for any $\alpha \in \F_{2^n}\setminus \{0\}$, we have $(\alpha x)^d = \alpha^d x^d$.

The \emph{binomial functions} $F\colon \F_{2^n} \rightarrow \F_{2^n}$ are those that can be written as $F(x) = x^a + \omega x^b$, where $a,b \in \mathbb{Z}$ and $\omega \in \F_{2^n}$. For special choices of $a$ and $b$, we can also easily give an LE-automorphism as follows.

\begin{proposition}
\label{prop:binom}
Let $F\colon \F_{2^n} \rightarrow \F_{2^n}, x \mapsto x^a+\omega x^b$ be a binomial function. Let $\alpha \in \F_{2^n}$ be an element of order $d$, where $d|(b-a)$. Then
\[\left[\begin{array}{cc}\alpha & 0\\ 0 & \alpha^a\end{array}\right] \in \AutLE(F).\]
In particular, if $\gcd(b-a,2^n-1)\neq 1$, the group $\AutLE(F)$ is non-trivial.
\end{proposition}
\begin{proof}
We have 
    $(\alpha x)^a + \omega (\alpha x)^b = \alpha^a x^a + \omega \alpha^b x^b = \alpha^a(x^a + \omega \alpha^{b-a} x^b)$, and $\alpha^{b-a} = 1$ because $\ord(\alpha)|(b-a)$. 
\end{proof}
    
    \begin{example}[APN function\footnote{recently classified into an infinite family, see~\cite{DBLP:journals/iacr/BudaghyanHK19}} defined in Theorem 2 of~\cite{DBLP:journals/tit/EdelKP06}]
        Let $u \in \F_{2^{10}}^*$ be an element of order 3. The function $F\colon \F_{2^{10}} \rightarrow \F_{2^{10}}, x \mapsto x^3 + \omega x^{36}$ is APN if and only if $\omega \in \{u\F_{2^5}^*\} \cup \{u^2\F_{2^5}^*\}$.
        
        Since  $\gcd(36-3,2^{10}-1) = 33$, a non-trivial automorphism in $\AutLE(F)$ can be given by an element of order $33$.  \qed
    \end{example}

\begin{example}[APN functions defined in Theorem 1 of~\cite{DBLP:journals/tit/BudaghyanCL08}]
\label{ex:cor1}
    Let $s$ and $k$ be positive integers with $\gcd(s,3k)=1$ and let $t \in \{1,2\}, i=3-t$. Let further $a = 2^s+1$ and $b = 2^{ik}+2^{tk+s}$ and let $\omega = \alpha^{2^k-1}$ for a primitive element $\alpha \in \F_{2^{3k}}^*$. If $\gcd(2^{3k}-1,(b-a)/(2^k-1)) \neq \gcd(2^k-1,(b-a)/(2^k-1))$, the function $F\colon \F_{2^{3k}} \rightarrow \F_{2^{3k}}, x \mapsto x^a + \omega x^b$ is APN.
    
    If $\gcd(2^{3k}-1,(b-a)/(2^k-1)) \neq 1$, the conditions of Proposition~\ref{prop:binom} are fulfilled and a non-trivial automorphism can be given by an element of order $\gcd(2^{3k}-1,(b-a))$.  Otherwise, $\gcd(2^k-1,(b-a)/(2^k-1)) \neq 1$ by assumption and, since $(2^k-1)(2^k+2^{2k}+1) = 2^{3k}-1$, also $\gcd(2^{3k}-1,(b-a)) \neq 1$. Therefore, a non-trivial automorphism in $\AutLE(F)$ always exists. \qed
\end{example}

\begin{example}[APN functions defined in Theorem 2 of~\cite{DBLP:journals/tit/BudaghyanCL08}]
    Let $s$ and $k$ be positive integers such that $s \leq 4k-1$, $\gcd(k,2)=gcd(s,2k)=1$, and $i=sk \mod 4$, $t=4-i$. Let further $a = 2^s+1$ and $b = 2^{ik}+2^{tk+s}$ and let $\omega = \alpha^{2^k-1}$ for a primitive element $\alpha \in \F_{2^{4k}}^*$. Then, the function $F\colon \F_{2^{4k}} \rightarrow \F_{2^{4k}}, x \mapsto x^a + \omega x^b$ is APN.
    
    We will show that $b-a \mod 5 = 0$. Then, since $(2^4-1)|(2^{4k}-1)$, a non-trivial automorphism can be given. Let $m\in \mathbb{N}$ be such that $i=4m+sk$. Indeed, the following equalities hold modulo $5$:
    \begin{align*}
        b-a &= 2^{ik}+2^{tk+s}-2^s-1 = 2^{(4m+sk)k} +2^{tk+s}-2^s-1 \\
        &= (2^{k^2})^s+2^{tk+s}-2^s-1 = 2^{tk+s}-1,
    \end{align*}
    where the last equality is fulfilled because $k$ is odd and thus, $2^{k^2} = 2 \mod 5$. It is left to show that $2^s2^{tk}= 1 \mod 5$. This can easily be obtained by considering the four different cases of $s \in \{1,3\} \mod 4$, $k \in \{1,3\} \mod 4$. In particular, we use the fact that $tk+s = (4-i)k+s$, which is equal to $(4-sk)k+s \mod 4$. \qed
\end{example}
    
\paragraph{Extension to Multinomial Functions} Proposition~\ref{prop:binom} can easily be generalized to multinomial functions as follows.
\begin{proposition}
\label{prop:multinom}
Let $F\colon \F_{2^n} \rightarrow \F_{2^n}, x \mapsto \sum_{i=0}^{k-1} \omega_i x^{a_i}$ with $\omega_i \in \F_{2^n}$ and $a_i \in \{0,\dots,2^n-1 \}$. Let $\alpha \in \F_{2^n}$ be an element of order $d$, such that, for all $i \in \{0,\dots,k-1\}$, $d|(a_i-a_0)$. Then, 
\[\left[\begin{array}{cc}\alpha & 0\\ 0 & \alpha^{a_0}\end{array}\right] \in \AutLE(F).\]
In particular, if $\gcd(a_1-a_0,a_2-a_0,\dots,a_{k-1}-a_0,2^n-1)\neq 1$, the group $\AutLE(F)$ is non-trivial.
\end{proposition}
\begin{proof}
We have 
    \[\sum_{i=0}^{k-1} \omega_i (\alpha x)^{a_i} = \sum_{i=0}^{k-1} \omega_i \alpha^{a_i} x^{a_i} = \alpha^{a_0} \sum_{i=0}^{k-1} \omega_i \alpha^{a_i-a_0} x^{a_i} = \alpha^{a_0} \sum_{i=0}^{k-1} \omega_i x^{a_i}.\]
\end{proof}

\begin{example}[APN functions defined in Theorem 1 of~\cite{DBLP:journals/ffa/BrackenBMM08}]
    Let $k$ and $s$ be odd integers with $\gcd(k,s) = 1$. Let $b \in \F_{2^{2k}}$ which is not a cube, $c \in \F_{2^{2k}} \setminus \F_{2^k}$, and, for $i \in \{1,\dots,k-1\}$, let $r_i \in \F_{2^k}$. Then, the function
    \[F\colon \F_{2^{2k}} \rightarrow \F_{2^{2k}}, \quad x \mapsto bx^{2^s+1} + b^{2^k}x^{2^{k+s}+2^k}+cx^{2^k+1}+\sum_{i=1}^{k-1}r_ix^{2^{i+k}+2^i}\]
    is APN.
    
    Recall that $2^{2k}-1 = 0 \mod 3$. To show that $F$ admits a non-trivial automorphism according to Proposition~\ref{prop:multinom}, we see that
    \renewcommand{\labelenumi}{(\roman{enumi})}
    \begin{enumerate}
        \item $2^{k+s}+2^k-2^s-1=0 \mod 3$, and
        \item for all $i \in \{0,1,\dots,k-1\}$, we have $2^{i+k}+2^i-2^s-1 = 0 \mod 3$. \qed
    \end{enumerate}
\end{example}

\begin{example}[APN functions defined in Theorem 2.1 of~\cite{DBLP:journals/ccds/BrackenBMM11}]
    Let $k$ and $s$ be positive integers such that $k+s = 0 \mod 3$ and $gcd(s,3k) = \gcd(3,k)=1$. Let further $u \in \F_{2^{3k}}^*$ be primitive and let $v,w \in \F_{2^k}$ with $vw \neq 1$. Then, the function 
    \[F\colon \F_{2^{3k}} \rightarrow \F_{2^{3k}}, \quad x \mapsto ux^{2^s+1}+u^{2^k}x^{2^{2k}+2^{k+s}}+vx^{2^{2k}+1}+wu^{2^k+1}x^{2^{k+s}+2^s}\]
    is APN.
    
    Recall that $2^{3k}-1 = 0 \mod 7$. To show that $F$ admits a non-trivial automorphism according to Proposition~\ref{prop:multinom}, we show that
    \renewcommand{\labelenumi}{(\roman{enumi})}
    \begin{enumerate}
        \item $2^{2k}+2^{k+s}-2^s-1 = 0 \mod 7$, and
        \item $2^{2k}+1-2^s-1 = (2^{k})^2-2^s = 0 \mod 7$, and
        \item $2^{k+s}+2^s-2^s-1 = 2^{k+s}-1 = 0 \mod 7$.
    \end{enumerate}
    Case $(iii)$ holds because $k+s=0 \mod 3$ and thus, $2^{k+s} = 1 \mod 7$. Case $(ii)$ can be deduced by considering the two cases of $k=1 \mod 3$ and $k = 2 \mod 3$ separately. Case $(i)$ immediately follows from $(ii)$ and $(iii)$. \qed
\end{example}
    
\subsection{Generalized Butterflies}
\label{sec:gen_butterfly}
It was shown in~\cite{DBLP:conf/crypto/PerrinUB16} that the sporadic APN permutation in dimension six (i.e., Dillon's permutation) can be decomposed into a special structure.
\begin{definition}[Generalized Butterfly~\cite{DBLP:journals/tit/CanteautDP17}]
Let $n \in \mathbb{N}$ be odd and let $R \colon \F_{2^n} \times \F_{2^n} \rightarrow \F_{2^n}$ be such that, for all fixed $y \in \F_{2^n}$, the mapping $x \mapsto R(x,y)$ is a permutation of $\F_{2^n}$. Then, an \emph{open generalized butterfly} is defined as a permutation 
\[\mathsf{H}_{R}\colon \F_{2^n} \times \F_{2^n} \rightarrow  \F_{2^n} \times \F_{2^n}, \quad (x,y) \mapsto \left(R(y,R^{-1}(x,y)),R^{-1}(x,y)\right).\]

If $R(x,y) = (x+\alpha y)^3 + \beta y^3$ for some $\alpha,\beta \in \F_{2^{n}}, \beta \neq 0$, $\mathsf{H}_R$ is called an \emph{open generalized butterfly with exponent 3} and it is denoted by $\mathsf{H}_{\alpha,\beta}$. 
\end{definition}

Dillon's permutation is affine equivalent to $\mathsf{H}_{\alpha,\beta}$ for  $n=3, \beta = 1$ and an $\alpha \in \F_{2^n}^*$ with $\tr(\alpha) = 0$. More generally, open generalized butterflies with exponent 3 are APN for $n=3, \alpha \neq 0$ with $\tr(\alpha)=0$ and $\beta \in \{\alpha^3+\alpha,\alpha^3+\alpha^{-1}\}$ and they are never APN for other odd values of $n$. Even when generalizing to Gold-like exponents, i.e., for $R(x,y) = (x+\alpha y)^{2^i+1}+ \beta y^{2^i+1}$ with $\gcd(i,n)=1$, the open generalized butterflies $\mathsf{H}_R$ are never APN for odd values of $n \neq 3$, see~\cite{butterfly_not_apn}.

Let $\zeta$ denote a non-zero element of the finite field $\F_{2^n}$. For 
\[A = \left[\begin{array}{cc}
    {\zeta^3} & 0 \\
     0 & {\zeta}
\end{array}\right],\]
we have $\mathsf{H}_{\alpha,\beta} \circ A = A \circ \mathsf{H}_{\alpha,\beta}$ for any $\alpha,\beta \in \F_{2^{n}}, \beta \neq 0$. Thus, there always exists a non-trivial element in $\AutLE(\mathsf{H}_{\alpha,\beta})$. 

For $n=3$, it is easy to verify that all matrices $A$ of the structure above for $\zeta \neq 1$ are similar to $\companion(X^6+X^5+X^4+X^3+X^2+X+1)$. This corresponds to Class 5 of the possible LE-automorphisms for 6-bit permutations (see Corollary~\ref{cor:classes}), which we deduce later.

\subsection{Known APN Functions in Small Dimension}
Up to dimension $n=5$, all APN functions are CCZ-equivalent to monomial functions, see~\cite{DBLP:journals/dcc/BrinkmannL08}. 

In~\cite{DBLP:journals/amco/EdelP09}, for dimensions $n \in \{6,7,8\}$, Edel and Pott constructed new APN functions up to CCZ-equivalence from the APN functions known at that time (see~\cite{banff}) by the so-called ``switching construction''. In particular, they listed 14 CCZ-inequivalent APN functions in dimension six, 19 in dimension seven, and 23 in dimension eight. All but one of them (see~\cite[Theorem 11]{DBLP:journals/amco/EdelP09}) are either monomial or quadratic functions. For the function inequivalent to monomial functions and quadratic functions, the authors of~\cite{DBLP:journals/amco/EdelP09} computed the order of the automorphism group,\footnote{In~\cite{DBLP:journals/amco/EdelP09}, this notion corresponds to the multiplier group $\mathcal{M}(G_F)$ of the development of $G_F$.} which is 8. Note that this is the only known example of an APN function that is not CCZ-equivalent to either a monomial function or a quadratic function. This function was discovered independently by Brinkmann and Leander in~\cite{DBLP:journals/dcc/BrinkmannL08}.  Note that in the recent works~\cite{cryptoeprint:2020:295,9000901}, the authors have found a new class of quadratic APN functions which lead to new APN functions in dimension 9. In~\cite{cryptoeprint:2020:295}, it was shown that some of the APN functions from~\cite{DBLP:journals/amco/EdelP09} can be classified into an infinite class. The 10-bit APN functions defined in~\cite{DBLP:journals/tit/EdelKP06} have recently been classified into an infinite family as well~\cite{DBLP:journals/iacr/BudaghyanHK19}.

In~\cite{matrix_full,DBLP:journals/dcc/YuWL14}, the authors found 471 new CCZ-inequivalent quadratic APN functions in dimension seven, and 8157 new CCZ-inequivalent quadratic APN functions in dimension eight. In~\cite{weng2013quadratic}, the authors presented 10 new CCZ-inequivalent quadratic APN functions in dimension eight.\footnote{In~\cite{weng2013quadratic}, the authors also claimed to have found 285 CCZ-inequivalent APN functions in dimension seven. However, their computational results are not available anymore, so we do not know whether those functions contain new APN functions.}  In~\cite{yu:eprint}, two new quadratic shift-invariant APN functions in dimension 9 have been found. We also refer to the Tables in~\url{https://boolean.h.uib.no/mediawiki/index.php/Tables} for an up-to-date list of the known APN functions.

To the best of our knowledge, the functions mentioned above include all the sporadic APN functions known at the time of submission of this manuscript (i.e., March 2020). 

\subsection{APN Permutations}
One of the most interesting problems in this area concerns the construction of APN permutations. For instance, it is natural to apply a permutation when substituting $n$-bit strings by other $n$-bit strings, as it is usually done in symmetric cryptographic algorithms. Then, APN permutations offer the best resistance against differential cryptanalysis. Note that the property of being a permutation is not invariant under CCZ-equivalence. Actually, not many examples of APN permutations are known. To the best of our knowledge, the CCZ-equivalence classes which are, at the time of submission of this manuscript, known to contain APN permutations fall in one of the following three cases, where the first two cases define infinite families of functions, and the third one is a sporadic example which is not classified into an infinite family yet.\footnote{In the GitHub repository~\cite{Perrin_sboxu}, Perrin implements an algorithm that checks whether an APN function is CCZ-equivalent to a permutation. We tested all cases of APN functions that come from infinite classes in dimension 7 and 9 (see the list in~\cite{thesis_sun} and the new class from~\cite{cryptoeprint:2020:295} and~\cite{9000901}), all the cases for dimension 7 listed in~\cite{DBLP:journals/amco/EdelP09}, and the two functions in dimension 9 found in~\cite{yu:eprint}. Besides the monomial functions, none of them are CCZ-equivalent to a permutation.}

\begin{enumerate}
    \item The CCZ-equivalence classes represented by APN monomial functions for $n$ odd.
    \item The CCZ-equivalence classes represented by the quadratic functions $F\colon \F_{2^{3k}} \mapsto \F_{2^{3k}}, x \mapsto x^{2^s+1} + \omega x^{2^{ik}+2^{tk+s}}$, where $s,k$ are positive integers with $k$ odd, $\gcd(k,3)=\gcd(s,3k)=1$, $i=sk \mod 3$, $t=3-i$, and $\omega \in \F_{2^{3k}}^*$ with order $2^{2k}+2^k+1$ (Corollary 1 of~\cite{DBLP:journals/tit/BudaghyanCL08}).
    \item The CCZ-equivalence class of the quadratic (non-invertible) function $F\colon \F_{2^6} \mapsto \F_{2^6}, x \mapsto x^3 + \alpha x^{24} + x^{10}$, where $\alpha$ is an element in $\F_{2^6}^*$ with minimal polynomial $X^6 + X^4 + X^3 + X + 1$ (i.e., the CCZ-equivalence class of Dillon's permutation~\cite{browning2010apn}).
\end{enumerate}

The authors of~\cite{DBLP:journals/tit/BudaghyanCL08} showed that the functions in Class 2 are CCZ-inequivalent to Gold functions when $k \geq 4$. In~\cite{Yoshiara2016}, Yoshiara proved that if a quadratic APN function is CCZ-equivalent to a monomial function, it must be EA-equivalent to a Gold function. Therefore, the functions in Class 2 are CCZ-inequivalent to any monomial function when $k\geq 4$.

For each of the different CCZ-equivalence classes coming from the above cases, one can give a representative which is a permutation and admits a non-trivial LE-automorphism. Indeed, the monomial functions are shift-invariant and a non-trivial LE-automorphism can be given as described in Section~\ref{sec:shift-inv}. Class 2 defines a special case of Theorem 1 in~\cite{DBLP:journals/tit/BudaghyanCL08}, which was covered in Example~\ref{ex:cor1} above. Finally, Class 3 is covered by the generalized butterfly structure with exponent 3, described in Section~\ref{sec:gen_butterfly}.

\paragraph{A Conjecture on the Automorphisms of APN Functions and Permutations} For all of the known APN functions, the automorphism group is non-trivial (recall that this is clear for functions CCZ-equivalent to quadratic and monomial functions and that the only known APN function which is not CCZ-equivalent to either a quadratic or a monomial function has an automorphism group of order 8). Moreover, for all known APN permutations, a CCZ-equivalent permutation $G$ can be given with $|\AutLE(G)| > 1$. Therefore, we formulate the following conjecture.
\begin{conjecture}
\label{conj}
For an APN function $F\colon \F_2^n \rightarrow \F_2^n$, we have $|\Aut(F)| > 1$. Moreover, if $F$ is an APN permutation, there exists a CCZ-equivalent permutation $G$ with $|\AutLE(G)| > 1$.
\end{conjecture}

In the spirit of this conjecture, we are going to search for APN permutations with non-trivial LE-automorphisms in small dimensions. In the remainder of this paper, we describe our method for carrying out this search.

\begin{remark}
Intuitively, the property of having a non-trivial automorphism group should be quite rare among all the $n$-bit to $n$-bit functions for a fixed dimension $n \geq 6$. We took 10,000 functions from $\F_2^6$ to $\F_2^6$ sampled uniformly at random and computed their automorphism groups. All of them were trivial. 
\end{remark}

\begin{remark}
The size of the group of LE-automorphisms is invariant under linear equivalence, but not under affine equivalence. In particular if $|\AutLE(F)| > 1$ for a function $F\colon \F_2^n \rightarrow \F_2^n$, then there might exist a constant $c \in \F_2^n$ such that $\AutLE(F+c)$ is trivial. For example, this is the case for the 6-bit APN permutation found by Algorithm~\ref{alg:search} (see Section~\ref{sec:res_6}).

More precisely, if $F(0) = 0$, one can show that $\AutLE(F+c)$ is a subgroup of $\AutLE(F)$ given by
\[ \AutLE(F+c) = \left\{ \sigma = \left[\begin{array}{cc}A & 0\\ 0 & B\end{array}\right] \mid \sigma \in \AutLE(F), Bc = c\right\}.\]
\end{remark}

\begin{remark}
One can further ask whether for any APN \emph{function}, there exists a representative in its CCZ-equivalence class which admits a non-trivial LE-automorphism. We checked that this is the case for any APN function in dimension $n \leq 5$. 
\end{remark}

\section{Equivalences for Permutations with Non-Trivial LE-Auto\-morphisms}
If we want to classify all $n$-bit permutations $F$ (up to CCZ-equivalence) with non-trivial elements \[\left[\begin{array}{cc}A & 0\\ 0 & B\end{array}\right] \in \AutLE(F),\] we can significantly reduce the number of tuples $(B,A)$ to consider. The observations in this section result in Corollary~\ref{cor:classes}, which states that for $n=6,7,8$ we only need to consider $17,27$, and $32$ tuples $(B,A)$, respectively. Note that this classification holds for \emph{any} permutation with a non trivial LE-automorphism, not only for APN permutations. 

Let $F \circ A = B \circ F $ for a \emph{function} $F \colon \F_2^n \rightarrow \F_2^n$ and $A,B \in \GL(n,\F_2)$. For a permutation $P$ on $\F_2^n$, we define the sets of points of order divisible by $i$ as $\Ord(P,i) \coloneqq \{x \in \F_2^n \mid P^i(x) = x\}$, which is a subspace of $\F_2^n$ if $P$ is linear.
\begin{lemma}
\label{lem:order}
Let $F \colon \F_2^n \rightarrow \F_2^n$ and let $A,B$ be permutations on $\F_2^n$ such that $F \circ A = B \circ F$. Then, for all $i \in \mathbb{N}$, \[x \in \Ord(A,i) \quad \Rightarrow \quad F(x) \in \Ord(B,i).\]
Moreover, if $F$ is a permutation, the above implication is an equivalence.\end{lemma} 
\begin{proof}
Observe that, for all $i \in \mathbb{N}$, $F \circ A^i = B^i \circ F$. If $x \in \Ord(A,i)$, then $F(x) = B^i(F(x))$, thus $F(x) \in \Ord(B,i)$. Let on the other hand $F(x) \in \Ord(B,i)$. Then, $F(x) = F(A^i (x))$. Thus, $x =A^i(x)$ if $F$ is a permutation.
\end{proof}

We only need to consider $A$ and $B$ of prime order, as shown in the following.
\begin{lemma}
Let $F \colon \F_2^n \rightarrow \F_2^n$ for which there exists a non-trivial automorphism in $\AutLE(F)$. Then, $F \circ A= B \circ F$ with $A,B \in \GL(n,\F_2)$ such that either 
\begin{enumerate}
    \item $\ord(A)=\ord(B)=p$ for $p$ prime, or
    \item $A = I_n$ and $\ord(B) = p$ for $p$ prime, or
    \item $B = I_n$ and $\ord(A) = p$ for $p$ prime.
\end{enumerate}
If $F$ is a permutation, the first of the above conditions must hold.
\end{lemma}
\begin{proof}
Let $g \in \AutLE(F)$, $g \neq I_n$. We consider the cyclic subgroup $\langle g \rangle \subseteq \AutLE(F)$. From the fundamental theorem of cyclic groups, it contains a cyclic subgroup of prime order. Let this subgroup be generated by
\[h = \left[\begin{array}{cc}A & 0\\ 0 & B\end{array}\right].\]
The result follows since $\ord(h) = \lcm(\ord(A),\ord(B))$. If $F$ is a permutation, then $\ord(A) = \ord(B)$ because of Lemma~\ref{lem:order}.
\end{proof}

Two matrices $M,M' \in \GL(n,\F_2)$ are called \emph{similar}, denoted $M \sim M'$, if there exists a matrix $P \in \GL(n,\F_2)$ such that $M' = P^{-1} M P$. It is obvious that similarity defines an equivalence relation. Moreover, we can provide a representative of each equivalence class as follows.

\begin{lemma}{(Rational Canonical Form)\cite[Page~476]{dummit1991}}
\label{lem:rcf}
Every matrix $M \in \GL(n,\F_2)$ is similar to a unique matrix $M' \in \GL(n,\F_2)$ of the form
\[ M' = \left[\begin{array}{cccc}
\companion(q_r) & & & \\
& \companion(q_{r-1}) & & \\
& & \ddots & \\
& & & \companion(q_1)
\end{array}\right]\]
for polynomials $q_i$ such that $q_r \mid q_{r-1} \mid \dots \mid q_1$. The matrix $M'$ in the form above is called the \emph{rational canonical form} of $M$, denoted $\rcf(M)$.
\end{lemma}

If $A' \sim A$, $B' \sim B$, and $F \circ A= B \circ F$, there exists a function $G$ which is linear equivalent to $F$ and for which $G \circ A' = B' \circ G$. Therefore, if we are only interested in $F$ up to linear equivalence, it is sufficient to consider $A$ and $B$ in rational canonical form. 

We can reduce the search space further if we use the fact that all powers of automorphisms are also automorphisms. Based on this fact, we consider the following equivalence classes for matrices of prime order.
\begin{definition}
Let $A,B,C,D \in \GL(n,\F_2)$ be of order $p$ for $p$ prime. The tuple $(A,B)$ is said to be \emph{power-similar} to the tuple $(C,D)$, denoted $(A,B) \sim_p (C,D)$, if there exists $i \in \mathbb{N}$ such that $A \sim C^i$ and $B \sim D^i$. The tuple $(A,B)$ is said to be \emph{extended power-similar} to $(C,D)$, denoted $(A,B) \sim_{Ep} (C,D)$, if one of the two following conditions holds:
\begin{enumerate}
    \item $(A,B) \sim_p (C,D)$
    \item $(A^{-1},B^{-1}) \sim_p (D,C)$.
\end{enumerate}
\end{definition}
Power-similarity and extended power-similarity define equivalence relations on the tuples of matrices in $\GL(n,\F_2^n)$ of the same prime order. Therefore, we can deduce the following lemma.

\begin{lemma}
Let $F \colon \F_2^n \rightarrow \F_2^n$ with an automorphism \[\left[\begin{array}{cc}A & 0\\ 0 & B\end{array}\right] \in \AutLE(F)\] for $A, B\in \GL(n,\F_2)$ being of prime order $p$. For every $(\widetilde{B},\widetilde{A})$ power-similar to $(B,A)$, there is a function $G$ linear-equivalent to $F$ such that \begin{equation}\label{eq:rcf_form}\left[\begin{array}{cc}\rcf(\widetilde{A}) & 0\\ 0 & \rcf(\widetilde{B})\end{array}\right] \in \AutLE(G).\end{equation} 

Moreover, if $F$ is a permutation, then for every tuple $(\widetilde{B},\widetilde{A})$ extended power-similar to $(B,A)$, there is such a function $G$ fulfilling Equation~(\ref{eq:rcf_form}) and being linear equivalent to either $F$ or $F^{-1}$.
\end{lemma}
\begin{proof}
Let $A = P^{-1} \widetilde{A}^i P$ and let $B = Q^{-1} \widetilde{B}^i Q$. We have that \[F \circ A = B \circ F \quad \Leftrightarrow \quad F \circ P^{-1} \widetilde{A}^i P = Q^{-1} \widetilde{B}^i Q \circ F \] and, thus, for $G \coloneqq Q \circ F \circ P^{-1}$, we have $G \circ \widetilde{A}^i = \widetilde{B}^i \circ G$. Let $k = i^{-1} \mod p$. Then,
\begin{align*}
    G \circ \widetilde{A}^{ik} = \widetilde{B}^{ik} \circ G \quad \Leftrightarrow \quad  G \circ \widetilde{A} = \widetilde{B} \circ G.
\end{align*}
Without loss of generality, $\widetilde{A}$ and $\widetilde{B}$ can be chosen up to similarity. If we chose them in rational canonical form, we obtain the first part of the lemma. 

The second part can be obtained by the same argument and using the fact that, for a permutation $F$, we have $F^{-1} \circ B^{-1} = A^{-1} \circ F^{-1}$.
\end{proof}

Thus, if we want to consider all permutations  with a non-trivial linear self-equivalence up to CCZ-equivalence (since $F$ is CCZ-equivalent to $F^{-1}$), we can restrict to a single tuple $(B,A)$ from each equivalence class under extended power-similarity.

Therefore, combining all the lemmas established in this section, we can enumerate the tuples we need to consider for $n \in \{6,7,8\}$ as follows. The code for generating all those tuples can be found in~\cite{cbe90}. 
\begin{corollary}
\label{cor:classes}
Let $n \in \{6,7,8\}$. All linear-equivalence classes of permutations $F \colon \F_2^n \rightarrow \F_2^n$ or $F^{-1}$ with a non-trivial automorphism 
\[\left[\begin{array}{cc}A & 0\\ 0 & B\end{array}\right] \in \AutLE(F)\]
can be obtained by considering the following classes for tuples $(B,A)$:

For $n=6$, we have the 17 classes
\begin{enumerate}
    \item $B = \companion(X^6+X^5+X^4+X^3+X+1) \quad A = \companion(X^6+X^3+X^2+1)$
    \item $B = \companion(X^6+X^5+X^4+X^3+X+1) \quad A = \companion(X^6+X^5+X^3+X^2+X+1)$
    \item $B = \companion(X^6+X^5+X^4+X^3+X+1) \quad A = \companion(X^6+X^5+X^4+1)$
    \item $B = A = \companion(X^6+X^5+X^4+X^3+X+1)$
    \item $B = A = \companion(X^6+X^5+X^4+X^3+X^2+X+1)$
    \item $B = A = I_1 \oplus \companion(X^5+1)$
    \item $B = I_2 \oplus \companion(X^4+X^3+X^2+1) \quad A = I_2 \oplus \companion(X^4+X^2+X+1)$
    \item $B = A = I_2 \oplus \companion(X^4+X^3+X^2+1)$
    \item $B = A = \companion(X^3+1) \oplus \companion(X^3+1)$
    \item $B = \companion(X^3+X^2+1) \oplus \companion(X^3+X^2+1) \quad A = \companion(X^6+X^5+X^4+X^3+X^2+X+1)$
    \item $B = \companion(X^3+X^2+1) \oplus \companion(X^3+X^2+1) \quad A = \companion(X^3+X+1) \oplus \companion(X^3+X+1)$
    \item $B = A = \companion(X^3+X^2+1) \oplus \companion(X^3+X^2+1)$
    \item $B = A = I_3 \oplus \companion(X^3+1)$
    \item $B = A = \companion(X^2+1) \oplus \companion(X^2+1) \oplus \companion(X^2+1)$
    \item $B = A = \companion(X^2+X+1) \oplus \companion(X^2+X+1) \oplus \companion(X^2+X+1)$
    \item $B = A = I_2 \oplus \companion(X^2+1) \oplus \companion(X^2+1)$
    \item $B = A = I_4  \oplus \companion(X^2+1)$.
\end{enumerate}

For $n=7$, we have the 27 classes
\begin{enumerate}
    \item $B = A = \companion(X^7+1)$
    \item $B = \companion(X^7+X^6+X^5+X^4+X^3+X^2+1) \quad A = \companion(X^7+X^3+1)$
    \item $B = \companion(X^7+X^6+X^5+X^4+X^3+X^2+1) \quad A = \companion(X^7+X^5+X^3+X+1)$
    \item $B = \companion(X^7+X^6+X^5+X^4+X^3+X^2+1) \quad A = \companion(X^7+X^5+X^4+X^3+X^2+X+1)$
    \item $B = \companion(X^7+X^6+X^5+X^4+X^3+X^2+1) \quad A = \companion(X^7+X^6+1)$
    \item $B = \companion(X^7+X^6+X^5+X^4+X^3+X^2+1) \quad A = \companion(X^7+X^6+X^4+X+1)$
    \item $B = \companion(X^7+X^6+X^5+X^4+X^3+X^2+1) \quad A = \companion(X^7+X^6+X^5+X^2+1)$
    \item $B = \companion(X^7+X^6+X^5+X^4+X^3+X^2+1) \quad A = \companion(X^7+X^6+X^5+X^3+X^2+X+1)$
    \item $B = \companion(X^7+X^6+X^5+X^4+X^3+X^2+1) \quad A = \companion(X^7+X^6+X^5+X^4+1)$
    \item $B = \companion(X^7+X^6+X^5+X^4+X^3+X^2+1) \quad A = \companion(X^7+X^6+X^5+X^4+X^2+X+1)$
    \item $B = A = \companion(X^7+X^6+X^5+X^4+X^3+X^2+1)$
    \item $B = I_1 \oplus \companion(X^6+X^5+X^4+X^3+X+1) \quad A = I_1 \oplus \companion(X^6+X^3+X^2+1)$
    \item $B = I_1 \oplus \companion(X^6+X^5+X^4+X^3+X+1) \quad A = I_1 \oplus \companion(X^6+X^5+X^3+X^2+X+1)$
    \item $B = I_1 \oplus \companion(X^6+X^5+X^4+X^3+X+1) \quad A = I_1 \oplus \companion(X^6+X^5+X^4+1)$
    \item $B = A = I_1 \oplus \companion(X^6+X^5+X^4+X^3+X+1)$
    \item $B = A = I_2 \oplus \companion(X^5+1)$
    \item $B = \companion(X^3+X+1) \oplus \companion(X^4+X^3+X^2+1) \quad A = \companion(X^7+1)$
    \item $B = \companion(X^3+X+1) \oplus \companion(X^4+X^3+X^2+1) \quad A = \companion(X^3+X^2+1) \oplus \companion(X^4+X^2+X+1)$
    \item $B = A = \companion(X^3+X+1) \oplus \companion(X^4+X^3+X^2+1)$
    \item $B = I_3 \oplus \companion(X^4+X^3+X^2+1) \quad A = I_3 \oplus \companion(X^4+X^2+X+1)$
    \item $B = A = I_3 \oplus \companion(X^4+X^3+X^2+1)$
    \item $B = A = I_1 \oplus \companion(X^3+1) \oplus \companion(X^3+1)$
    \item $B = A = \companion(X^2+X+1) \oplus \companion(X^2+X+1) \oplus \companion(X^3+1)$
    \item $B = A = I_4 \oplus \companion(X^3+1)$
    \item $B = A = I_1 \oplus \companion(X^2+1) \oplus \companion(X^2+1) \oplus \companion(X^2+1)$
    \item $B = A = I_3 \oplus \companion(X^2+1) \oplus \companion(X^2+1)$
    \item $B = A = I_5 \oplus \companion(X^2+1)$.
\end{enumerate}

For $n=8$, we have the 32 classes 
\begin{enumerate}
    \item $B = \companion(X^8+X^7+X^6+X^4+X^2+X+1) \quad A = \companion(X^8+X^5+X^4+X^3+1)$
    \item $B = A = \companion(X^8+X^7+X^6+X^4+X^2+X+1)$
    \item $B = \companion(X^8+X^7+X^6+X^5+X^4+X^3+X^2+1) \quad A = \companion(X^8+X^6+X^4+X^3+X^2+1)$
    \item $B = \companion(X^8+X^7+X^6+X^5+X^4+X^3+X^2+1) \quad A = \companion(X^8+X^6+X^5+X^4+X^2+1)$
    \item $B = \companion(X^8+X^7+X^6+X^5+X^4+X^3+X^2+1) \quad A = \companion(X^8+X^6+X^5+X^4+X^3+X^2+X+1)$
    \item $B = \companion(X^8+X^7+X^6+X^5+X^4+X^3+X^2+1) \quad A = \companion(X^8+X^7+X^4+X^3+X+1)$
    \item $B = \companion(X^8+X^7+X^6+X^5+X^4+X^3+X^2+1) \quad A = \companion(X^8+X^7+X^5+X^2+X+1)$
    \item $B = \companion(X^8+X^7+X^6+X^5+X^4+X^3+X^2+1) \quad A = \companion(X^8+X^7+X^5+X^4+X+1)$
    \item $B = \companion(X^8+X^7+X^6+X^5+X^4+X^3+X^2+1) \quad A = \companion(X^8+X^7+X^6+1)$
    \item $B = \companion(X^8+X^7+X^6+X^5+X^4+X^3+X^2+1) \quad A = \companion(X^8+X^7+X^6+X^3+X+1)$
    \item $B = \companion(X^8+X^7+X^6+X^5+X^4+X^3+X^2+1) \quad A = \companion(X^8+X^7+X^6+X^5+X^3+1)$
    \item $B = A = \companion(X^8+X^7+X^6+X^5+X^4+X^3+X^2+1)$
    \item $B = A = I_1 \oplus \companion(X^7+1)$
    \item $B = I_2 \oplus \companion(X^6+X^5+X^4+X^3+X+1) \quad A = I_2 \oplus \companion(X^6+X^3+X^2+1)$
    \item $B = I_2 \oplus \companion(X^6+X^5+X^4+X^3+X+1) \quad A = I_2 \oplus \companion(X^6+X^5+X^3+X^2+X+1)$
    \item $B = I_2 \oplus \companion(X^6+X^5+X^4+X^3+X+1) \quad A = I_2 \oplus \companion(X^6+X^5+X^4+1)$
    \item $B = A = I_2 \oplus \companion(X^6+X^5+X^4+X^3+X+1)$
    \item $B = A = I_3 \oplus \companion(X^5+1)$
    \item $B = \companion(X^4+X^3+X^2+1) \oplus \companion(X^4+X^3+X^2+1) \quad A = I_1 \oplus \companion(X^7+1)$
    \item $B = \companion(X^4+X^3+X^2+1) \oplus \companion(X^4+X^3+X^2+1) \quad A = \companion(X^4+X^2+X+1) \oplus \companion(X^4+X^2+X+1)$
    \item $B = A = \companion(X^4+X^3+X^2+1) \oplus \companion(X^4+X^3+X^2+1)$
    \item $B = A = \companion(X^4+X^3+X^2+X+1) \oplus \companion(X^4+X^3+X^2+X+1)$
    \item $B = I_4 \oplus \companion(X^4+X^3+X^2+1) \quad A = I_4 \oplus \companion(X^4+X^2+X+1)$
    \item $B = A = I_4 \oplus \companion(X^4+X^3+X^2+1)$
    \item $B = A = \companion(X^2+X+1) \oplus \companion(X^3+1) \oplus \companion(X^3+1)$
    \item $B = A = I_2 \oplus \companion(X^3+1) \oplus \companion(X^3+1)$
    \item $B = A = I_5 \oplus \companion(X^3+1)$
    \item $B = A = \companion(X^2+1) \oplus \companion(X^2+1) \oplus \companion(X^2+1) \oplus \companion(X^2+1)$
    \item $B = A = \companion(X^2+X+1) \oplus \companion(X^2+X+1) \oplus \companion(X^2+X+1) \oplus \companion(X^2+X+1)$
    \item $B = A = I_2 \oplus \companion(X^2+1) \oplus \companion(X^2+1) \oplus \companion(X^2+1)$
    \item $B = A = I_4 \oplus \companion(X^2+1) \oplus \companion(X^2+1)$
    \item $B = A = I_6 \oplus \companion(X^2+1)$.
\end{enumerate}
\end{corollary}
Note that, when searching for permutations in one of the above classes, we can reduce the search space further by filtering candidates up to affine equivalence as follows. For a matrix $M \in \GL(n,\F_2)$, let $\comm(M)$ denote the subgroup of $\GL(n,\F_2)$ of all matrices that commute with $M$.
\begin{lemma}
\label{lem:comm}
Let $F \colon \F_2^n \rightarrow \F_2^n$ and let $A,B \in \GL(n,\F_2)$ such that $F \circ A = B \circ F$. Then, $G \circ A = B \circ G$ for any $G = C_B \circ F \circ C_A$ with $C_A \in \comm(A)$ and $C_B \in \comm(B)$.
\end{lemma}
\begin{proof}
We have $G \circ A = C_B \circ F \circ C_A \circ A = C_B \circ F \circ A \circ C_A = C_B \circ B \circ F \circ C_A = B \circ C_B \circ F \circ C_A = B \circ G$.
\end{proof}

This allows us to only consider one representative within the class of $\comm(B) \circ F \circ \comm(A)$.

\section{Searching for APN Permutations with Non-Trivial LE-Automorphisms}
We now use the classification of tuples from Corollary~\ref{cor:classes} to search for APN permutations with non-trivial linear self-equivalences in dimensions $n \in \{6,7,8\}$. First, we observe that not all of the tuples $(B,A)$ obtained in Corollary~\ref{cor:classes} have to be considered in the search. 

\begin{definition}
Let $A,B \in \GL(n,\F_2)$. A tuple $(B,A)$ is \emph{admissible} if there exists an APN permutation $F \colon \F_2^n \rightarrow \F_2^n$ with $F \circ A = B \circ F$. 
\end{definition}

The following two propositions provide necessary conditions for a tuple to be admissible.

\begin{proposition}
\label{prop:inv}
Let $F \colon \F_2^n \rightarrow \F_2^n$ be an APN permutation. Let $A_1,A_2 \subseteq \F_2^n$ be two affine subspaces of $\F_2^n$ such that $F(A_1) = A_2$. Then, $\dim A_i \notin \{2,4,n-1\}$.
\end{proposition}
\begin{proof}
Let $d = \dim A_1 = \dim A_2$. Without loss of generality, we can choose $A_1 = A_2 = \F_2^d \times \{0\}^{n-d}$ by considering a permutation $F'$ that is affine equivalent to $F$. Because $F'$ is APN, the property $F'(A_1) = A_2$ implies the existence of an APN permutation in dimension $d$. This cannot happen for $d=2$ and $d=4$, see~\cite{DBLP:journals/dam/Hou06}. The case of $d=n-1$ was shown in \cite[Proposition 2.1]{DBLP:journals/dam/Hou06}. 
\end{proof}

Therefore, if $(B,A) \in \GL(n,\F_2) \times \GL(n,\F_2)$ is an admissible tuple we have, for all $i \in \mathbb{N}$, $\dim \Ord(A,i) = \dim \Ord(B,i) \notin \{2,4,n-1\}$.

\begin{proposition}
\label{prop:poly}
Let $(B,A)  \in \GL(n,\F_2) \times \GL(n,\F_2)$ be an admissible tuple with $\ord(A) = \ord(B) = k$ for $k$ prime. Then, there exists no quadrinomial in $\F_2[X]/(X^k+1)$ that is a multiple of both the minimal polynomial of $A$ and the minimal polynomial of $B$.
\end{proposition}
\begin{proof}
Suppose there is such a quadrinomial $p = X^a + X^b + X^c + 1$. Then, both $A^a + A^b = A^c +1$ and $B^a + B^b = B^c +1$ hold. Let $g \in \F_{2}^n$ be an element with $\ord_A(g) = k$ 
and let $F$ be a permutation that fulfills the self-equivalence for $(B,A)$. We have
\begin{align*}
    F(A^ag) + F(A^bg) = B^aF(g) + B^bF(g) &= (B^a+B^b)F(g)\\
    &= (B^c+1)F(g) = F(A^cg) + F(g),
\end{align*}
which implies that $F$ cannot be APN.
\end{proof}

\paragraph{Depth-First Tree Search Algorithm}
Using the above propositions, for some of the tuples given in Corollary 1 we can immediately see that they are not admissible. For the other tuples, we can check their admissibility by the recursive depth-first tree search described in Algorithm~\ref{alg:search} at the end of the paper. 

The algorithm receives two matrices $A,B \in \GL(n,\F_2)$ as input and constructs the look-up tables of all $n$-bit APN permutations $F$ with $F \circ A = B \circ F$ up to linear equivalence. To reduce the search space according to Lemma~\ref{lem:comm}, we provide a subset of $C_A \subseteq \comm(A)$ and a subset $C_B \subseteq \comm(B)$ as additional inputs. The idea is that the algorithm should output only the \emph{smallest} representative of an APN permutation up to conjugation with elements in $C_A$, resp. $C_B$, where the term \emph{smallest} refers to some lexicographic ordering of the look-up table (this check is performed by the procedure $\textsc{isSmallest}$).\footnote{The test for being the smallest representative is omitted if the depth exceeds some threshold $t$, because at some depth it is faster to just traverse the remaining tree. Therefore, it might happen that the algorithm outputs more representatives than just the smallest.}

At the beginning of the search, the look-up table is initialized to $\bot$ at each entry, where $\bot$ indicates that the respective entry is not yet defined. At the beginning of each iteration of the recursive $\textsc{NextVal}$ procedure, the algorithm checks whether the look-up table is completely defined, i.e., whether there are no entries marked $\bot$ left (procedure $\textsc{isComplete}$). If it is completely defined, the look-up table is appended to the list of solutions and the procedure returns. Otherwise, the algorithm chooses the next undefined entry $x$ in the look-up table according to some previously defined ordering (procedure $\textsc{NextFreePosition}$) and sets $F(x)$ to the next value $y$ (also according to some previously defined ordering) which does not yet occur in the look-up table and for which $\ord_A(x) = \ord_B(y)$. Besides fixing $F(x) \coloneqq y$, the algorithm further fixes $F(A^i(x)) \coloneqq B^i(y)$ for all $i$ according to the self-equivalence. For each entry that is fixed, the algorithm checks whether the partially-defined function can still be APN (procedure $\textsc{IsAPN}$). In case that the APN property has already been violated,  the entry $x$ is set to the next possible value $y$. In case that the partially-defined function can still be APN, the algorithm goes one level deeper.

\paragraph{APN Check}
We would like to stress that it is important to implement Algorithm~\ref{alg:search} carefully in order to obtain the results presented in this work. Indeed, the description above is a bit simplified. Instead of checking for each newly fixed entry whether the partially-defined function can still be APN (for example by constructing a partial difference distribution table (DDT)), we use a global two-dimensional array, called $\mathrm{ddt}$, and initialized to $0$, of size $2^n \times 2^n$ that dynamically stores the partial DDT. Recall that the DDT of a function $F \colon \F_2^n \rightarrow \F_2^n$ is defined as the $2^n \times 2^n$ integer matrix containing $|\{ x \in \F_2^n \mid F(x) + F(x+\alpha) = \beta\}|$ at the position in row $\alpha$ and column $\beta$. After each entry of $F$ is fixed, we update the partial DDT according to the newly fixed entry and check whether, for any $\alpha \neq 0$, it contains values larger than $2$ (in that case, $F$ cannot be APN). This is done by the following procedure, where $\oplus$ denotes the bitwise XOR operation and $\wt(\alpha)$ denotes the Hamming weight of the integer $\alpha$ interpreted as an element in $\F_2^n$:
\begin{algorithmic}[1]
		\Function{addDDTInformation}{c}
		\For {$\alpha \in [1,\dots,2^n-1]$ with $\wt(\alpha)$ being even}
		    \If {$\mathrm{sbox}[c \oplus \alpha] \neq \bot$}
		        \State $\mathrm{ddt}[\alpha][\mathrm{sbox}[c]\oplus\mathrm{sbox}[c\oplus\alpha]] \gets \mathrm{ddt}[\alpha][\mathrm{sbox}[c]\oplus\mathrm{sbox}[c\oplus\alpha]] + 2$
		        \If{$\mathrm{ddt}[\alpha][\mathrm{sbox}[c]\oplus\mathrm{sbox}[c\oplus\alpha]] > 2$}
		            \State \Return 0
		        \EndIf
		    \EndIf
		\EndFor
		\State \Return 1
		\EndFunction
\end{algorithmic}
		
The input parameter $c$ corresponds to the position that we fix in $\mathrm{sbox}$. If an entry is reset to $\bot$ (line 32 of Algorithm~\ref{alg:search}) the $\mathrm{ddt}$ array also has to be updated accordingly, which is implemented by the following procedure. Note that in the for loop, the elements $\alpha$ have to be traversed in the same order as in $\textsc{addDDTInformation}$. The parameter $c$ corresponds to the entry that is reset to $\bot$.

\begin{algorithmic}[1]
		\Function{removeDDTInformation}{c}
		\For {$\alpha \in [1,\dots,2^n-1]$ with $\wt(\alpha)$ being even}
		    \If {$\mathrm{sbox}[c \oplus \alpha] \neq \bot$}
		        \State $\mathrm{ddt}[\alpha][\mathrm{sbox}[c]\oplus\mathrm{sbox}[c\oplus\alpha]] \gets \mathrm{ddt}[\alpha][\mathrm{sbox}[c]\oplus\mathrm{sbox}[c\oplus\alpha]] - 2$
		        \If {$\mathrm{ddt}[\alpha][\mathrm{sbox}[c]\oplus\mathrm{sbox}[c\oplus\alpha]] = 2$}
		            \State \Return 0
		        \EndIf
		    \EndIf
		\EndFor
		\EndFunction
	\end{algorithmic}

Note that we do not have to consider the whole DDT for the APN check. As it was shown in~\cite[Theorem 4]{DBLP:conf/eurocrypt/BethD93}, only those DDT entries corresponding to input differences with even Hamming weight have to be computed. 

For the details of our implementation, we refer to the source code provided in~\cite{cbe90}.
\renewcommand{\algorithmicrequire}{\textbf{Input:}}
\renewcommand{\algorithmicensure}{\textbf{Output:}}
\algnewcommand{\algorithmicgoto}{\textbf{go to}}
\algnewcommand{\Goto}[1]{\algorithmicgoto~\ref{#1}}
\begin{algorithm}
	\caption{\textsc{APNsearch}} \label{alg:search}
	\begin{algorithmic}[1]
		\Require Matrices $A, B \in \GL(n,\F_2)$, $C_A \subseteq \comm(A), C_B \subseteq \comm(B)$. Global array $\mathrm{sbox}$ of size $2^n$, initialized to $\mathrm{sbox}[i] = \bot$, for all $i \in \{0,\dots,2^n-1\}$
		\Ensure All $n$-bit APN permutations $F$ s.t.\ $FA = BF$ up to linear equivalence.
		    \vspace{.1em}
		    \State $L \gets \{ \}, \quad \mathrm{sbox}[0] \gets 0$
		    \State $\textsc{nextVal}(0)$
            \State \Return $L$ 
            \vspace{1em}
		\Function{nextVal}{depth}
		    \If{$\textsc{isComplete}(\mathrm{sbox})$} 
		        \State $L \gets L \cup \{\mathrm{sbox}$\}
		        \State \Return
		    \EndIf
		    \State $x \gets \textsc{nextFreePosition}()$
		    \For{$y \in \F_2^n$}
		        \If{$y$ is not in $\mathrm{sbox}$ and $\ord_A(x) = \ord_B(y)$} \Comment{$x$ interpreted to be in $\F_2^n$}
		            \State $xS \gets x, \quad yS \gets y$
		            \For{$i=0$ to $\ord_A(x)-2$}
		                \State $\mathrm{sbox}[xS] \gets yS$
		                \If{not $\textsc{isAPN}(\mathrm{sbox})$}
		                    \State \Goto{undo}
		                \EndIf
		                \State $xS \gets A(xS), \quad yS \gets B(yS)$
		            \EndFor
		            \State $\mathrm{sbox}[xS] \gets yS$
		            \If{not $\textsc{isAPN}(\mathrm{sbox})$}
		                \State \Goto{undo}    
		            \EndIf
		            \If{$\mathrm{depth} \leq t$}
		                 \If{$\textsc{isSmallest}(\mathrm{sbox})$} 
		                 \State $\textsc{nextVal}(\mathrm{depth}+1)$
		            \EndIf
		            \Else \State $\textsc{nextVal}(\mathrm{depth}+1)$
		            \EndIf
		            \Repeat \label{undo}
		                \State $\mathrm{sbox}[xS] \gets \bot$
		                \State $xS \gets A^{-1}(xS)$
		            \Until{$xS = A^{-1}(x)$}
		        \EndIf
		    \EndFor
		\EndFunction
	\end{algorithmic}
\end{algorithm}

\begin{table}[ht!]
\caption{Analysis of the tuples $(B,A)$ given in Corollary~\ref{cor:classes}. ``No.''\@ corresponds to the number of the tuple in Corollary~\ref{cor:classes}. 
The column ``admissible'' indicates whether there exists an $n$-bit APN permutation $F$ for which $F \circ A = B \circ F$. In case that it does, we list \emph{all} the solutions for the CCZ-equivalence classes of such $F$. The ``?'' indicates that we were not able to either exclude the tuple by Prop.~\ref{prop:inv} or \ref{prop:poly}, or to finish the exhaustive search for $F$.\label{tab:results}}
\centering
{\small
\begin{tabular}{|lll|lll|ll|}
    \hline
     \multicolumn{3}{|c|}{$n=6$} & \multicolumn{3}{|c|}{$n=7$} & \multicolumn{2}{|c|}{$n=8$}\\
     No. & admissible & solutions & No. & admissible & solutions & No. & admissible\\
     \hline
     1 & no (Alg.~\ref{alg:search}) &  & 1 & yes & $x \mapsto x^{5}$ & 1 & no (Alg.~\ref{alg:search}) \\
      &  &  &  &  & $x \mapsto x^{9}$ &  &  \\
       &  &  &  &  & $x \mapsto x^{63}$ &  &  \\
        &  &  &  &  & $x \mapsto x^{78}$ &  &  \\
         &  &  &  &  & $x \mapsto x^{85}$ &  &  \\
          &  &  &  &  & $x \mapsto x^{88}$ &  &  \\
     2 & no (Alg.~\ref{alg:search}) &  & 2 & no (Prop.~\ref{prop:poly}) &  & 2 & no (Alg.~\ref{alg:search}) \\
     3 & no (Alg.~\ref{alg:search}) &  & 3 & no (Prop.~\ref{prop:poly}) &  & 3 & no (Prop.~\ref{prop:poly}) \\
     4 & no (Prop.~\ref{prop:poly}) &  & 4 & yes & $x \mapsto x^{63}$ & 4 & no (Alg.~\ref{alg:search}) \\
     5 & yes & Dillon's~\cite{browning2010apn} & 5 & yes & $x \mapsto x^{9}$ & 5 & no (Alg.~\ref{alg:search})\\
     6 & no (Prop.~\ref{prop:inv}) &  & 6 & no (Prop.~\ref{prop:poly}) &  & 6 & no (Alg.~\ref{alg:search}) \\
     7 & no (Alg.~\ref{alg:search}) &  & 7 & yes & $x \mapsto x^{5}$ & 7 & no (Prop.~\ref{prop:poly}) \\
     8 & no (Prop.~\ref{prop:poly}) &  & 8 & yes & $x \mapsto x^{78}$ & 8 & no (Alg.~\ref{alg:search}) \\
     9 & no (Prop.~\ref{prop:inv}) &  & 9 & yes & $x \mapsto x^{85}$ & 9 & no (Alg.~\ref{alg:search}) \\
     10 & no (Alg.~\ref{alg:search}) &  & 10 & yes & $x \mapsto x^{88}$ & 10 & no (Alg.~\ref{alg:search}) \\
     11 & no (Alg.~\ref{alg:search}) &  & 11 & no (Prop.~\ref{prop:poly}) &  & 11 & no (Prop.~\ref{prop:poly})\\
     12 & no (Prop.~\ref{prop:poly}) &  & 12 & no (Prop.~\ref{prop:inv}) & & 12 & no (Prop.~\ref{prop:poly}) \\
     13 & no (Prop.~\ref{prop:inv}) &  & 13 & no (Prop.~\ref{prop:inv}) &  & 13 & no (Prop.~\ref{prop:inv}) \\
     14 & no (Alg.~\ref{alg:search}) &  & 14 & no (Prop.~\ref{prop:inv}) &  & 14 & no (Alg.~\ref{alg:search}) \\
     15 & no (Alg.~\ref{alg:search}) &  & 15 & no (Prop.~\ref{prop:inv}) &  & 15 & no (Alg.~\ref{alg:search}) \\
     16 & no  (Prop.~\ref{prop:inv}) &  & 16 & ? &  & 16 & no (Alg.~\ref{alg:search}) \\
     17 & no (Prop.~\ref{prop:inv}) &  & 17 & no (Alg.~\ref{alg:search}) &  & 17 & no (Prop.~\ref{prop:poly}) \\
      &  &  & 18 & no (Alg.~\ref{alg:search}) &  & 18 & no (Prop.~\ref{prop:inv}) \\
      &  &  & 19 & no (Prop.~\ref{prop:poly}) &  & 19 & no (Prop.~\ref{prop:inv}) \\
      &  &  & 20 & no (Prop.~\ref{prop:inv}) &  & 20 & no (Prop.~\ref{prop:inv})\\
      &  &  & 21 & no (Prop.~\ref{prop:inv}) &  & 21 & no (Prop.~\ref{prop:inv}) \\
      &  &  & 22 & ? &  & 22 & ? \\
      &  &  & 23 & ? &  & 23 & no (Alg.~\ref{alg:search}) \\
      &  &  & 24 & no (Alg.~\ref{alg:search}) &  & 24 & no (Prop.~\ref{prop:poly}) \\
      &  &  & 25 & no (Prop.~\ref{prop:inv}) &  & 25 & no (Prop.~\ref{prop:inv}) \\
      &  &  & 26 & no (Alg.~\ref{alg:search}) &  & 26 & no (Prop.~\ref{prop:inv}) \\
      &  &  & 27 & no (Prop.~\ref{prop:inv}) &  & 27 & no (Alg.~\ref{alg:search}) \\
      &  &  &  &  &  & 28 & no (Prop.~\ref{prop:inv}) \\
      &  &  &  &  &  & 29 & no (Alg.~\ref{alg:search}) \\
      &  &  &  &  &  & 30 & ? \\
      &  &  &  &  &  & 31 & no (Alg.~\ref{alg:search})  \\
      &  &  &  &  &  & 32 & no (Prop.~\ref{prop:inv}) \\
      \hline
\end{tabular}
}
\end{table}

\subsection{Results for $n=6$}
\label{sec:res_6}
By Propositions~\ref{prop:inv} and~\ref{prop:poly}, we immediately obtain that $8$ out of the $17$ tuples given in Corollary~\ref{cor:classes} are not admissible (see Table~\ref{tab:results}). We performed an \emph{exhaustive} search for APN permutations in the remaining $9$ tuples using Algorithm~\ref{alg:search}. The case of \begin{equation}\label{eq:complicated} B= A = \companion(X^2+1) \oplus \companion(X^2+1) \oplus \companion(X^2+1)\end{equation}
(Class 14) was the most difficult one. The functions $A$ and $B$ are involutions and therefore only consist of cycles of length 1 and 2, which causes Algorithm~\ref{alg:search} to be less efficient. However, without loss of generality we can set $F$ on the fixed points of $A$ to an APN permutation. More precisely, let $F\colon \F_2^n \rightarrow \F_2^n$ be an APN permutation with $F \circ A = B \circ F$ for $A,B \in \GL(n,\F_2
^n)$ with $\dim \Ord(A,1) = \dim \Ord(B,1) = k$. Let further $\pi_A \colon \F_2^k \rightarrow \Ord(A,1)$ and $\pi_B \colon \F_2^k \rightarrow \Ord(B,1)$ be (vector space) isomorphisms. Then there exists an APN permutation $G \colon \F_2^k \rightarrow \F_2^k$ such that, for all $x \in \Ord(A,1)$, we have $F(x) = \pi_B(G(\pi_A^{-1}(x)))$. The following definition and lemma show a sufficient condition that allows to fix this APN permutation $G$ up to affine equivalence.
\begin{definition}
An element $A \in \GL(n,\F_2)$ is called \emph{extendable} if, for all linear permutations $L \colon \Ord(A,1) \rightarrow \Ord(A,1)$, there exists an element $L' \in \comm(A)$ such that the restriction of $L'$ to $\Ord(A,1)$ equals $L$.
\end{definition}

\begin{lemma}
Let $A,B \in \GL(n,\F_2)$ be extendable and let $F\colon \F_2^n \rightarrow \F_2^n$ be a permutation for which $F \circ A = B \circ F$. Further, for all $x \in \Ord(A,1)$, let $F(x) = \pi_B(G(\pi_A^{-1}(x)))$, where  $G \colon \F_2^{k} \rightarrow \F_2^{k}$ with $k = \dim \Ord(A,1)$ and $\pi_A \colon \F_2^k \rightarrow \Ord(A,1)$, $\pi_B \colon \F_2^k \rightarrow \Ord(B,1)$ are vector space isomorphisms. For every permutation $G'$ affine-equivalent to $G$, there exists a permutation $F'$ affine-equivalent to $F$ such that $F' \circ A = B \circ F'$ and, for all $x \in \Ord(A,1)$, it fulfills
\[F'(x) = 
    \pi_B(G'(\pi_A^{-1}(x))).\]
\end{lemma}
\begin{proof}
Let  $C\colon \F_2^k \rightarrow \F_2^k, x \mapsto Lx + c$ be an affine permutation with $L \in \GL(k,\F_2)$ and $c \in \F_2^k$. We first show the statement for $G' = G \circ C$ and then for $G' = C \circ G$, which finally proves the statement for all $G'$ affine-equivalent to $G$.

Case 1: $G' = G \circ C$. Because $A$ is extendable, there exists an affine permutation $C'\colon \F_2^n \rightarrow \F_2^n, x \mapsto L'x + \pi_A(c)$, where the linear permutation $L' \in \GL(n,\F_2)$ has the property that
\[ L'(x) = 
\pi_A(L(\pi_A^{-1}(x)))  \text{ if } x \in \Ord(A,1)\]
and fulfills $L'A = AL'$. Therefore, $C'A = L'A + \pi_A(c) = AL' + \pi_A(c) = A(L'+ \pi_A(c)) = AC'$, since $\pi_A(c) \in \Ord(A,1)$. For $x \in \Ord(A,1)$, we have
\[ F(C'(x)) = \pi_B(G(\pi_A^{-1}(L'x + \pi_A(c))))   = \pi_B(G(L\pi_A^{-1}(x) + c)) = \pi_B(G'(\pi_A^{-1}(x))).\]
Further, $F \circ C' \circ A = F \circ A \circ C' = B \circ F \circ C'$.

Case 2: $G' = C \circ G$. Because $B$ is extendable, there exists an affine permutation $C'\colon \F_2^n \rightarrow \F_2^n, x \mapsto L'x + \pi_B(c)$, where the linear permutation $L' \in \GL(n,\F_2)$ has the property that
\[ L'(x) = 
\pi_B(L(\pi_B^{-1}(x)))  \text{ if } x \in \Ord(B,1)\]
and fulfills $L'B = BL'$.  For $x \in \Ord(A,1)$, we have
\[ C'(F(x)) = L'(\pi_B(G(\pi_A^{-1}(x)))) + \pi_B(c) = \pi_B(LG(\pi_A^{-1}(x))+c) = \pi_B(G'(\pi_A^{-1}(x))).\]
Further, $C' \circ F \circ A = C' \circ B \circ F = B \circ C' \circ F$.
\end{proof}

For the case of $A = B$ as in Equation~(\ref{eq:complicated}), we checked that $A$, resp. $B$ is extendable. Therefore,
since $A$, resp. $B$ has eight fixed points and since there is only one APN permutation on $3$ bits up to affine equivalence, eight entries of $F$ can be fixed in advance. This trick reduced the computation time for this case to around 8 hours on a PC. 

To summarize, APN permutations exist only in the case 
\begin{equation} B = A = \companion(X^6+X^5+X^4+X^3+X^2+X+1)\end{equation}
(Class 5) and they are all CCZ-equivalent to Dillon's permutation.  As a conclusion, we have shown the following.
\begin{theorem}
Up to CCZ-equivalence, there is only one APN permutation $F$ in dimension 6 with $|\AutLE(F)| > 1$.
\end{theorem}

\subsection{Results for $n=7$}
By Propositions~\ref{prop:inv} and~\ref{prop:poly}, we directly obtain that $13$ out of the $27$ tuples given in Corollary~\ref{cor:classes} are not admissible. We performed an exhaustive search for APN permutations in 11 of the remaining $14$ tuples using Algorithm~\ref{alg:search}.

Class 1 corresponds to the shift-invariant permutations and obviously contains all the monomial permutations. By letting Algorithm~\ref{alg:search} run for several days on a cluster with 256 cores, we were able to finish the search for APN permutations in this class. We obtained that the APN monomial permutations are (up to CCZ-equivalence) the \emph{only} shift-invariant APN permutations in  this dimension (see Table~\ref{tab:results}). 
\begin{theorem}
Up to CCZ-equivalence, a shift-invariant APN permutation in dimension $7$ must be a monomial function.
\end{theorem}

The 6 APN monomial permutations in dimension 7 are also contained in Classes 4, 5, 7, 8, 9, and 10, respectively. Those classes correspond to tuples $(B,A)$, where $A$ and $B$ correspond to multiplications by elements in the finite field $\F_{2^7}$. We did not find any other APN permutations. This allows us to state the following theorem, which summarizes the three missing cases that are infeasible to handle with Algorithm~\ref{alg:search}.

\begin{theorem}
\label{thm:res7}
Let $F$ be an APN permutation in dimension 7 with $|\AutLE(F)| > 1$ that is not CCZ-equivalent to a monomial function. Then, $F$ is CCZ-equivalent to a permutation $G$ for which $G \circ A = B \circ G$ with
\begin{enumerate}
    \item $B = A = I_2 \oplus \companion(X^5+1)$ \quad or
    \item $B = A = I_1 \oplus \companion(X^3+1) \oplus \companion(X^3+1)$ \quad or
    \item $B = A = \companion(X^2+X+1) \oplus \companion(X^2+X+1) \oplus \companion(X^3+1)$.
\end{enumerate}
\end{theorem}

\subsection{Results for $n=8$}
By Propositions~\ref{prop:inv} and~\ref{prop:poly}, we directly obtain that $15$ out of the $32$ tuples given in Corollary~\ref{cor:classes} are not admissible. We performed an exhaustive search for APN permutations in 15 of the remaining $17$ tuples using Algorithm~\ref{alg:search}. We did not find any APN permutation.  To conclude, we state the following theorem. 

\begin{theorem}
\label{thm:res8}
Let $F$ be an APN permutation in dimension 8 with $|\AutLE(F)| > 1$. Then, $F$ is CCZ-equivalent to a permutation $G$ for which $G \circ A = B \circ G$ with
    \begin{enumerate}
        \item $B = A = \companion(X^4+X^3+X^2+X+1) \oplus \companion(X^4+X^3+X^2+X+1)$ \quad or
        \item $B = A = I_2 \oplus \companion(X^2+1) \oplus \companion(X^2+1) \oplus \companion(X^2+1)$.
    \end{enumerate}
\end{theorem}

Table~\ref{tab:results} summarizes our results. The source code of our implementation of Algorithm~\ref{alg:search} can be found in~\cite{cbe90}. For checking whether the solutions that we find are CCZ-equivalent to an already known APN permutation, we used the equivalent condition on code equivalence as explained in \cite{browning2010apn}. Practically, we used the code equivalence algorithm of the computer algebra system Magma \cite{MR1484478}, which for $n=7$ takes a few seconds on a PC. 

\subsection{Randomized Search}
In case that the search space for a tuple $(B,A)$ is so large that handling it with Algorithm~\ref{alg:search} would be infeasible, we can perform a random search for APN permutations $F$ for which $F \circ A = B \circ F$. It is straightforward to implement a randomized version of Algorithm 1. For that, before the initial call of $\textsc{NextVal}$, we randomly shuffle the order in which the values for $y$ are iterated in line 10. We abort the search after a predetermined amount of time and repeat with a new initial shuffling. Furthermore, since we are not aiming for an \emph{exhaustive} search, we omit the check for the smallest representative, i.e., set $t$ to $-1$.

We applied the randomized search for the 5 tuples for which we were not able to finish the exhaustive search, i.e., Classes 16, 22, and 23 for $n=7$, and Classes 22 and 30 for $n=8$. We did not find any APN permutation by letting the algorithm run for at least 128 CPU days for each of those cases.

\section{Conclusion and Open Questions}
We observed that all APN permutations known from the literature contain a permutation in their CCZ-equivalence class that admit a non-trivial linear self-equivalence. We performed an exhaustive search for $6$-bit APN permutations with such non-trivial linear self-equivalences and a partial search in dimension 7 and 8. 

We expect that there are no more APN permutations with non-trivial linear self-equivalences in dimension 7 and 8. As open problems, it would be interesting to settle the cases described in Theorems~\ref{thm:res7} and~\ref{thm:res8}, i.e., to show that those cases contain no APN permutations. Another (very ambitious) open problem is to prove or disprove Conjecture~\ref{conj}. This would certainly be considered a major breakthrough in the theory of APN functions.

\subsection*{Acknowledgment}
We thank the anonymous reviewers for their detailed and helpful comments. We further thank Anne Canteaut, Yann Rotella and Cihangir Tezcan for fruitful discussions at an early stage of this project.

\end{document}